\documentclass[11pt]{article}
\usepackage{graphicx,xspace,amsmath,amssymb,fullpage,amsthm}

\usepackage{algorithmic,algorithm}
\usepackage[usenames,dvipsnames]{color}

\newtheorem{theorem}{Theorem}
\newtheorem{lemma}{Lemma}
\newtheorem{corollary}{Corollary}

\newtheorem{proposition}{Proposition}

\newtheorem{example}{Example}

\newcommand{\IP}{\textsf{IP2}\xspace}
\newcommand{\BIP}{\textsf{BIP2}\xspace}
\newcommand{\BIPLP}{\textsf{BIP2 Above LP}\xspace}
\newcommand{\VC}{\textsf{Vertex Cover}\xspace}
\newcommand{\VCLP}{\textsf{Vertex Cover Above LP}\xspace}
\newcommand{\OCT}{\textsf{Odd Cycle Transversal}\xspace}
\newcommand{\ASAT}{\textsf{Almost 2-SAT}\xspace}

\newcommand{\NMC}{\textsf{Node Multiway Cut}\xspace}
\newcommand{\NMCLP}{\textsf{Node Multiway Cut Above LP}\xspace}
\newcommand{\MINSAT}{\textsf{Min SAT}\xspace}
\newcommand{\bbN}{\mathbb{N}}
\newcommand{\bbR}{\mathbb{R}}
\newcommand{\opt}{\mathbf{opt}}
\newcommand{\lp}{\mathbf{lp}}
\newcommand{\val}{\mathbf{val}}
\newcommand{\gap}{\mathbf{gap}}
\newcommand{\calI}{\mathcal{I}}

\begin{document}
\title{Linear-Time FPT Algorithms via Network Flow}
\author{Yoichi Iwata\thanks{Department of Computer Science,
    Graduate School of Information Science and Technology,
    The University of Tokyo.
    Research Fellow of Japan Society for the Promotion of Science.
    \texttt{y.iwata@is.s.u-tokyo.ac.jp}
  }
  \and
  Keigo Oka\thanks{Department of Computer Science,
    Graduate School of Information Science and Technology,
    The University of Tokyo.
    JST, ERATO, Kawarabayashi Large Graph Project.
    \texttt{ogiekako@is.s.u-tokyo.ac.jp}
  }
  \and
  Yuichi Yoshida\thanks{
    National Institute of Informatics, and Preferred Infrastructure, Inc.
    \texttt{yyoshida@nii.ac.jp}
    Supported by JSPS Grant-in-Aid for Research Activity Start-up (24800082), MEXT Grant-in-Aid for Scientific Research on Innovative Areas (24106001), and JST, ERATO, Kawarabayashi Large Graph Project.
  }
}
\date{}
\maketitle

\begin{abstract}
In the area of parameterized complexity, to cope with NP-Hard problems,
we introduce a parameter $k$ besides the input size $n$,
and we aim to design algorithms (called FPT algorithms) that run in $O(f(k)n^d)$ time for some function $f(k)$ and constant $d$.
Though FPT algorithms have been successfully designed for many problems,
typically they are not sufficiently fast because of huge $f(k)$ and $d$.

In this paper, we give FPT algorithms with small $f(k)$ and $d$ for many important problems
including \OCT and \ASAT.
More specifically, we can choose $f(k)$ as a single exponential ($4^k$) and $d$ as one, that is,
linear in the input size.
To the best of our knowledge, 
our algorithms achieve linear time complexity for the first time for these problems.

To obtain our algorithms for these problems,
we consider a large class of integer programs, called \BIP.
Then we show that, in linear time, we can reduce \BIP to \VCLP preserving the parameter $k$, and we can compute an optimal LP solution for \VCLP using network flow.

Then, we perform an exaustive search by fixing half-integral values in the optimal LP solution for \VCLP.
A bottleneck here is that we need to recompute an LP optimal solution after branching.
To address this issue, we exploit network flow to update the optimal LP solution in linear time.

\end{abstract}

\thispagestyle{empty}
\setcounter{page}{0}

\section{Introduction}\label{sec:intro}
Assuming $\textrm{P} \neq \textrm{NP}$, there are no polynomial-time algorithms for NP-Hard problems in the worst case.
However, since many important problems are actually NP-Hard,
it is natural to study in which case these problems become polynomial-time tractable.
\emph{Parameterized complexity} is one of such direction, in which we measure the time complexity of algorithms with respect to the input size and another parameter.
A problem is called \emph{fixed parameter tractable (FPT)} with respect
to a parameter $k$ if it can be solved in time $f(k)n^{O(1)}$, where $n$ is the input size and $f$ is some computable function.
See~\cite{Downey:2012vk,Flum:2010vj,Niedermeier:2006wh} for books for parameterized complexity.

Though the initial motivation of parameterized complexity is making NP-Hard problems more tractable,
unfortunately many FPT algorithms have a disadvantage in their time complexity.
For example, the function $f(k)$ might be an astronomical tower of exponentials such as $2^{2^k}$ or the
degree $d$ of the polynomial in $n$ might be quite huge such as $n^{10}$.
Thus, it is desirable to improve these FPT algorithms so that $f(k)$ and $d$ become small \emph{simultaneously},
and our main contribution in this paper is giving FPT algorithms for many important problems whose running time is $O(c^k n)$ for some constant $c>1$.
That is, the polynomial part is only \emph{linear} whereas the function $f(k)$ is only a single exponential.

To describe the problems we are concerned with, we need to review previous results.
Indeed, there have been many studies on FPT algorithms with small $f(k)$ and $d$.
However, we note that many works pursuing small $d$ often neglect how large $f(k)$ is,
and many works pursuing small $f(k)$ often neglect how large $d$ is.

\OCT is a problem of finding the minimum vertex set whose removal makes the input graph to be bipartite.
We use the size of the optimal solution as a parameter.
Reed~et~al.~\cite{Reed:2004wb} proved that the problem is FPT by introducing a technique called \emph{iterative compression}.
The running time of their algorithm is $O(3^k knm)$.
Based on the graph minor theory, Fiorini~et~al.~\cite{fioriniHRV08} improved the polynomial part to be linear
for planar graphs.
For general graphs, Kawarabayashi and Reed~\cite{Kawarabayashi:2010uv} developed an
$O(f(k)(n+m)\alpha(n+m))$-time algorithm, where $\alpha(\cdot)$ denotes the inverse of the Ackermann function.

In \textsf{Max 2-SAT}, given a CNF with each clause containing two literals, we want to find an assignment so as to maximize the number of satisfied clauses.
\ASAT is a parameterized version of \textsf{Max 2-SAT}, in which a parameter is the minimum number of unsatisfied clauses.
Razgon and O'Sullivan~\cite{Razgon:2009tn} proved that \ASAT is FPT by designing an $O(15^k k m^3)$-time algorithm,
where $m$ is the number of clauses.
Their algorithm is also based on iterative compression.
Raman~et~al.~\cite{Raman:2011uj} improved the function $f(k)$ to $9^k$ by reducing the problem to \VC parameterized by the difference between the size of the optimal solution and the size of the maximum matching.
Here, \VC is the problem of finding a minimum set of vertices $S$ in a graph so that every edge is incident to $S$.
Cygan~et~al.~\cite{Cygan:2012fu} further improved $f(k)$ to $4^k$ by reducing the problem to \NMC,
which is a problem of finding a small subset of vertices whose removal
makes a given set of terminals separated, parameterized by the difference between the size of the
optimal solution and the LP lower bound.
And finally, Lokshtanov~et~al.~\cite{Lokshtanov:2012ud} obtained an $O^*(2.3146^k)$ algorithm by reducing the problem to \VC parameterized by the difference between the size of the optimal solution and the LP lower bound (we call the problem \VCLP).
They also showed that many other problems such as \OCT can be reduced to \VCLP and obtained faster algorithms.

Our first result is generalizing~\cite{Lokshtanov:2012ud} so that we can handle any problem in Binarized \IP.
Binarized \IP (\BIP in short) is a class of integer programs introduced by Hochbaum~\cite{Hochbaum:1993ww}.
In an instance of \BIP, each constraint has the form of $a_{i,j}x_i+b_{i,j}x_j+z_{i,j}\geq c_{i,j}$, where
$a_{i,j}$, $b_{i,j}$, and $c_{i,j}$ are integer constants with $a_{i,j},b_{i,j}\in\{-1,0,1\}$,
$x_i$ and $x_j$ are variables that can freely appear in other constraints, and $z_{i,j}$ is a
variable that can appear only in this constraint.
The objective function is a non-negative linear function of variables.
The variables can take any non-negative integer values.
We say that a variable is \emph{binary} if its domain is $\{0,1\}$, and 
if all variables in \BIP are binary, we call it as \emph{Binary \BIP}.
See Section~\ref{sec:preliminaries} for the formal definition of \BIP.
Many important problems including \VC, \ASAT, and \OCT can be easily written in the form of
(Binary) \BIP.

For an instance $\calI$ of $\BIP$, we define its \emph{integrality gap} as $\gap(\calI) = \opt(\calI) - \lp(\calI)$,
where $\opt(\calI)$ is the optimal IP value and $\lp(\calI)$ is the optimal LP value of its LP relaxation.
Then, \BIPLP is the problem of finding an optimal IP solution for an instance $\calI$ parameterized by $\gap(\calI)$.
Note that \VCLP is a special case of \BIPLP.
We give a generic procedure that reduces \BIPLP to \VCLP in linear time.
Note that instances that contain non-binary variables can also be reduced to \VCLP.
To this end, we use the property that any instance of \BIP admits half-integral optimal
primal/dual solutions~\cite{Hochbaum:2002eh}.
More specifically, we show the following.

\begin{theorem}\label{thm:reduction}
  Let $\calI$ be an instance of \BIP with $n$ variables and $m$ constraints.
  Suppose we have a pair of half-integral optimal primal/dual solutions for $\calI$.
  Then we can construct an instance $G$ of \VC of $O(n+m)$ vertices and edges with $\gap(G) =\gap(\calI)$ in $O(n+m)$ time,
  Furthermore, we can compute half-integral optimal primal/dual solutions for $G$ in $O(n+m)$ time.
\end{theorem}

Our second and main result is showing an algorithm that,
given half-integral optimal primal/dual solutions,
solves \VCLP in linear time.
\begin{theorem}\label{thm:linear}
  Suppose that we are given an instance $G=(V,E)$ of \VC and its half-integral optimal primal/dual solutions.
  Then, we can solve \VCLP in $O(4^k (|V|+|E|))$ time.
\end{theorem}
A key ingredient in our algorithm is network flow.
During exhaustive search,
the instance gradually changes and we need to update the optimal LP solution.
To avoid computing it from scratch,
we express the LP solution as a flow.
Then for each branching, in linear time, we update the flow and extract the optimal LP solution for the resulting instance from the flow.
We give a detailed explanation later.

From Theorem~\ref{thm:reduction}, we obtain the following corollary.
\begin{corollary}\label{cor:linear1}
  Given an instance $\calI$ of \BIP with $n$ variables and $m$ constraints and its half-integral optimal primal/dual solutions,
  we can solve \BIPLP in $O(4^k (n+m))$ time.
\end{corollary}

In order to obtain linear-time FPT algorithms from Corollary~\ref{cor:linear1}, the only remaining
part is to compute half-integral optimal primal/dual solutions in linear time.
Hochbaum~\cite{Hochbaum:2002eh} showed that for Binary \BIP,
we can compute them in $O(\ell(n+m))$ time, where $\ell$ is the optimal LP value.
Thus, by applying Corollary~\ref{cor:linear1} with these half-integral solutions,
we obtain the following.
Note that we can assume $\ell \leq k$.
\begin{theorem}\label{thr:linear2}
  We can solve \BIP in $O(4^k (n+m))$ time,
  where $n$ is the number of variables and $m$ is the number constraints.
\end{theorem}

As we have mentioned, Theorem~\ref{thr:linear2} immediately implies linear-time FPT algorithms for
\OCT and \ASAT.
To the best of our knowledge, they are first linear-time FPT algorithms for these problems.
To show the generality of Theorem~\ref{thr:linear2}, we give a list of other problems that can be
written in the form of Binary \BIP in Appendix~\ref{sec:problems}.

We mention here that, independently of our work, Ramanujan and Saurabh~\cite{Ramanujan13} also have obtained
an $O(4^k k^4 n)$-time algorithm for \ASAT via a different approach.

\paragraph{Proof Sketch.}
We now give a proof sketch of Theorem~\ref{thm:linear}.
Let $G=(V,E)$ be a \VC instance and $y^*$ be its half-integral optimal dual solution.
Then, we construct a network $\overline{G}$ from $G$ and a flow $f^*$ from $y^*$.
These network and flow will play a central role in our algorithm.
We also construct an optimal primal solution $x^*$ from the residual graph $\overline{G}_{f^*}$ of $\overline{G}$ with respect to $f^*$.
Using $x^*$, we can fix variables $v$ with $x^*_v = 0$ or $x^*_v = 1$.
Further, we want to find a set of variables $S$ that can be assigned integers without changing $\lp(G)$.
To this end, we will show that it suffices to find a strongly connected component $\overline{S}$ in $\overline{G}_{f^*}$ with a certain property.
We now assign integers to the corresponding set $S$ and this operation causes removal of $\overline{S}$ from $\overline{G}$.

Again, we want to find other variables that can be assigned integers.
Naively speaking, we can do so by recomputing a half-integral optimal dual solution, a feasible flow, the residual graph, and strongly connected components.
However, from the property of the network $\overline{G}$, we can reuse almost all information in the previous step.
It turns out that we can assign integers by keep finding and removing strongly connected components with a certain property from $\overline{G}_{f^*}$.
We can also keep the optimal flow $f^*$ at hand, and this process takes only linear time in total.

After this preprocess, the unique primal solution of $G = (V,E)$ is the one consisting of $\frac{1}{2}$ only.
Then, we can follow a standard approach to design an FPT algorithm.
That is, we pick an arbitrary edge $\{u,v\} \in E$ that is not covered yet by a variable assigned $1$, and we invoke recursion after assigning $x_u = 1$ or $x_v = 1$. 
To continue the process, instead of updating $y^*$, we directly update $f^*$ by augmenting paths.

We can show that $k$ decreases by $\frac{\Delta}{2}$ if the number of augmenting paths is $\Delta$.
Since other parts takes only linear time, 
the total running time becomes $O(2^{2k}(|V|+|E|)) = O(4^k(|V|+|E|))$.

\paragraph{Related Work:}
Before Lokshtanov~et~al.~\cite{Lokshtanov:2012ud} gave parameterized reductions from \OCT and \ASAT
to \VCLP, Cygan~et~al.~\cite{Cygan:2012fu} had given parameterized reductions to \NMCLP.
The reductions are one-way and it seems very difficult to express \NMC as \BIP.
This is because \VCLP and thus \BIP are known to have polynomial kernels~\cite{6375323} whereas obtaining
polynomial kernel for \NMC remains a (famous) open problem.
To the best of our knowledge, there is no linear-time FPT algorithm for \NMC in the literature,
and the fastest one (in terms of the input size) is $O(4^k \ell n^3)$-time algorithm
by Chen~et~al.~\cite{Chen:2007ha}, where $\ell$ is the number of terminals.
We note that, D\'aniel Marx mentioned that a linear-time FPT algorithm for \NMC is a folklore in his
lecture talk at GRAPH CUTS workshop in 2013.
However, we show an FPT algorithm for \NMC that runs in
$O(4^k (|V|+|E|))$-time by exploiting the network flow approach to make the result more accessible and show the general applicability of our approach (Appendix~\ref{sec:multiway}).
Note that by the reduction from \OCT or \ASAT to \NMCLP, though the difference between
the size of the optimal solution and the LP lower bound does not change, the optimal solution size itself increases.
Thus, a linear-time FPT algorithm for \NMC does not imply linear-time FPT algorithms for \OCT and \ASAT.

Compare to works pursuing small $f(k)$, 
there are a fewer number of works that pursue small degree $d$ of the polynomial.
However, several important results are known.
Bodlaender~\cite{Bodlaender96} developed a linear time FPT algorithm for computing
tree-decompositions of width $k$, and Courcelle~\cite{Courcelle90} showed that any graph property that can be expressed in
monadic second-order logic can be tested in linear time for graphs of fixed tree-width.
Kawarabayashi and Reed~\cite{kawarabayashi07linear} showed that the crossing number $k$ of a graph can be computed in linear time for a fixed $k$,
where \emph{crossing number} is the minimum number of crossings of edges when we draw the graph on the plane.
Fomin~et~al.~\cite{FominLMS12} developed a single exponential and linear time FPT
algorithm for \textsf{Planar-\textit{F} Deletion}.
Marx~et~al.~\cite{Marx11} showed that we can reduce the treewidth of a graph while preserving all the minimal $(s,t)$-separators of size up to $k$ in linear time for a fixed $k$, and obtained faster algorithms for several problems including \textsf{Bipartite Contraction}.
However, they concluded that their approach is very difficult to be generalized to other
problems such as \ASAT.

\paragraph{Organization:}
We give definitions used in this paper and introduce \BIP in Section~\ref{sec:preliminaries}.
In Section~\ref{sec:bip2}, we give a linear-time FPT algorithm for \VCLP and prove
Theorem~\ref{thm:linear}.
In Appendix~\ref{sec:problems}, we give a list of problems that can be expressed as Binary \BIP.
We give a proof of Theorem~\ref{thm:reduction} in Appendix~\ref{sec:reduction}.
In Appendix~\ref{sec:multiway}, we show a linear-time FPT algorithm for \NMC.

\section{Preliminaries}\label{sec:preliminaries}
In this section, we give definitions used in this paper and introduce binarized integer programming (\BIP).
We denote a set of non-negative integers by $\bbN$, and a set of non-negative half-integers by
$\bbN_{1/2}$, where a \emph{half-integer} is a multiple of $\frac{1}{2}$.

Let $G=(V,E)$ be an undirected graph.
The \emph{neighborhood} $N(v)$ of a vertex $v$ is $\{u\in V\mid \{u,v\}\in E\}$, and the
\emph{neighborhood} $N(S)$ of a vertex set $S\subseteq V$ is $\bigcup_{v\in S}N(v)\setminus S$.
We denote a set of edges incident to a vertex $v$ by $\delta(v)$.
For a subset $S\subseteq V$,
let $G[S]$ denote the \emph{subgraph induced by $S$}.

For a directed graph $G=(V,E)$ and its vertex $v\in V$, we denote the set of incoming edges of $v$ by
$\delta^-(v)$ and the set of outgoing edges of $v$ by $\delta^+(v)$.
Similarly, for a vertex set $S$, we define $\delta^-(S)$ as the set of incoming edges of $S$ and $\delta^+(S)$ as the set of outgoing edges of $S$.
We say that a vertex set $S \subseteq V$ is a \emph{strongly connected component} if, for any two vertices $u,v\in S$, there is a directed path from $u$ to $v$.
We say that $S$ is a \emph{tail strongly connected component} if it is strongly connected and $\delta^+(S) = \emptyset$.
A directed graph is called \emph{strongly connected} if the whole vertex set $V$ is strongly connected.
It is known that we can compute strongly connected components in $O(|V|+|E|)$ time.

A \emph{weighted graph} is a pair of a graph $G=(V,E)$ and a function $w:V \to \bbN$.
For a vertex set $S\subseteq V$, we denote the sum $\sum_{v\in S}w(v)$ by $w(S)$.
Similarly, an \emph{edge-weighted graph} is a pair of a graph $G=(V,E)$ and a function $w:E \to \bbN$.

A \emph{network} is a pair of a directed graph $G=(V,E)$ and an edge capacity
function $c:E \to \bbN$.
For vertices $s,t\in V$, an \emph{$s$-$t$ flow} is a function $f:E \to
\bbN$ satisfying that, for some $F \geq 0$, $\sum_{e\in\delta^+(s)}f(e)=\sum_{e\in\delta^-(t)}f(e)=F$,
$\sum_{e\in\delta^+(v)}f(e)=\sum_{e\in\delta^-(v)}f(e)$ for all $v\in V\setminus\{s,t\}$, and $0\leq
f(e)\leq c(e)$ for all $e\in E$.
We call $F$ the \emph{amount} of the flow.
Given a flow $f$ and a set of vertices $S$,
the \emph{out-flow} of $S$ refers to the restriction of $f$ to edges $e \in \delta^+(S)$ with $f(e) > 0$.
and the \emph{in-flow} of $S$ refers to the restriction of $f$ to edges $e \in \delta^-(S)$ with $f(e) > 0$.
A \emph{residual graph} of a network $(G,c)$ with respect to its flow $f$ is a directed graph $G_f=(V,E_f)$ with
$E_f=\{(u,v)\mid ((u,v) \in E\text{ and } f(u,v)<c(u,v))\text{ or }((v,u) \in E\text{ and } 0<f(v,u)) \}$.
It is known that we can compute an $s$-$t$ flow of an amount $F$ (if exists) in $O(F(|V|+|E|))$ time
using the Ford-Fulkerson algorithm.

Now, we define a class of integer programs, called \IP, introduced by
Hochbaum~\cite{Hochbaum:2002eh,Hochbaum:1993ww}.
An instance of \IP is of the following form:
\begin{align*}
\begin{array}{lll}
\text{minimize }& \displaystyle \sum_{i \in V} w_i x_i + \sum_{(i,j)\in E_1}d_{i,j} z_{i,j}\\
\text{subject to }& a_{i,j}x_i+b_{i,j}x_j+z_{i,j}\geq c_{i,j}&\text{for }(i,j)\in
E_1,\\
&a_{i,j}x_i+b_{i,j}x_j\geq c_{i,j}& \text{for }(i,j)\in E_2,\\
& x_i \in \bbN & \text{for }i \in V, \\
& z_{i,j} \in \bbN & \text{for }(i,j)\in E_1.
\end{array}
\end{align*}
Here $E_1,E_2 \subseteq V \times V$ are sets of pairs,
$x_{i}$ and $z_{i,j}$ are non-negative integer variables,
$w_i$ and $d_{i,j}$ are non-negative integers, and $a_{i,j}$, $b_{i,j}$, and $c_{i,j}$ are integers.
We call the variables $x_{i,j}$, which can appear in several constraints, \textit{shared variables}, and call the variables $z_{i,j}$, which can appear in only one constraint, \textit{independent variables}.
In \IP, each constraint imposed on $E_1$ or $E_2$ can contain only two shared variables and at most one independent variable.
When all coefficients of every constraint, $a_{i,j}$ and $b_{i,j}$, are from $\{-1,
0, 1\}$, the problem is called \textit{binarized \IP} (\BIP in short).
Additionally, when all variables of \BIP are binary, i.e., they are from $\{0,1\}$, the problem is called \textit{Binary \BIP}.

For an IP instance $\calI$, we denote by $\opt(\calI)$ the optimal IP value.
Its LP relaxation can be obtained by replacing the constraint of the form $x \in \bbN$ by a constraint of the form $x \geq 0$ for every variable $x$.
For an IP instance $\calI$, we denote by $\lp(\calI)$ the optimal value of its LP relaxation.
For a (primal or dual) solution $x$, $\val(\calI,x)$ denotes the LP value obtained by $x$.
We define the \emph{integrality gap} of $\calI$ as $\gap(\calI) = \opt(\calI) - \lp(\calI)$.\footnote{In the context of approximation algorithms, the integrality gap is often defined as the \emph{ratio} of $\lp(\calI)$ to $\opt(\calI)$}

The LP relaxation of \BIP and its dual LP admit half-integral optimal solutions~\cite{Hochbaum:2002eh}.
Moreover, for Binary \BIP, we can compute the optimal LP solutions both in the primal problem and the dual problem in $O(F (n+m))$ time,
where $F$ is the optimal LP value~\cite{Hochbaum:2002eh}.
Many important problems can be formulated as \BIP.
Some examples are given in Appendix~\ref{sec:problems}.

\section{A Linear-Time FPT Algorithm for \VCLP}\label{sec:bip2}
In this section, we give an $O(4^k (|V|+|E|))$-time algorithm for \VCLP and
prove Theorem~\ref{thm:linear}.
Given a graph $G=(V,E)$, the (primal) LP relaxation of \VC can be written as follows:
\begin{align*}
  \textbf{(Primal-VC)} & \quad
  \begin{array}{lll}
    \text{minimize }&\displaystyle \sum_{v\in V}w(v)x_v\\
    \text{subject to }&x_v+x_u\geq 1 & \text{for }\{u,v\}\in E,\\
    & x_u \geq 0 & \text{for } u \in V. \\
  \end{array}
\end{align*}
The \emph{all-half vector} refers to a vector $x \in \bbR^V$ such that $x_v = \frac{1}{2}$ for every $v \in V$.
We write $x \equiv \frac{1}{2}$ if $x \in \bbR^V$ is the all-half vector.
\begin{lemma}[\cite{Nemhauser:1975ko}]\label{lem:Nemhauser}
  The primal LP relaxation of \VC satisfies the following.
  \begin{itemize}
  \setlength{\itemsep}{0pt}
  \item[(1)] It admits a half-integral optimal solution.
  \item[(2)] For any optimal LP solution $x^L \in \bbR^V$, there is an optimal integer solution $x^I \in \{0,1\}^V$ such that $x^I_v=x^L_v$ holds for every $v \in V$ for which $x^L_v$ is an integer.
  \item[(3)] If the all-half vector is an optimal LP solution, then $w(S)\leq w(N(S))$ holds for every independent set $S\subseteq V$.
  \item[(4)] If the all-half vector is an optimal LP solution and $w(S)=w(N(S))$ holds for some independent set $S\subseteq V$, the following $x \in \bbN^V_{1/2}$ is also an optimal LP solution:
  If $v \in S$, then $x_v = 0$.
  If $v \in N(S)$, then $x_v = 1$.
  Otherwise, $x_v = \frac{1}{2}$.
  \item[(5)] If $w(S)<w(N(S))$ for every independent set $S$, the all-half vector is the unique optimal LP solution.
  \end{itemize}
\end{lemma}

By using Properties~(1) and~(2), it is not hard to design an FPT algorithm by naive exhaustive search.
First, as long as the optimal LP solution $x^L \in \bbR^V$ contains a variable $v$ such that $x^L_v$ is an integer, we fix the value of $v$ as $x^L_v$.
Then, for each vertex $v\in V$, we try to fix $x^L_v=1$ and check whether the optimal LP value
increases.
If it remains the same, we can fix the value of $v$ to $1$, and if it increases, we restore
the value of $v$ to $\frac{1}{2}$.
By checking every vertex, the all-half vector becomes the unique optimal LP solution.
Now we pick an arbitrary edge $\{u,v\} \in E$ that is not covered yet by a variable whose value is fixed to $1$, and we invoke recursion after setting $x_u=1$ or $x_v=1$.
Since the all-half vector is the unique optimal LP solution, the value of the optimal LP solution
increases by $\Delta\geq \frac{1}{2}$.
Thus, by setting $x_u = 1$ or $x_v=1$, we can decrease $k$ by $\Delta$.
Hence, the depth of the search tree is bounded by $2k$.
In each node in the search tree, we might need to solve the LP relaxation $n$ times.
Let $T(n)$ be the running time to solve the LP relaxation, then the total running time is $O(4^k n
T(n))$, which is a huge polynomial.

The outline of our algorithm is similar to the naive exhaustive search.
However, we exploit all the properties in Lemma~\ref{lem:Nemhauser} to improve the running time to be linear.
We introduce the dual LP relaxation of \VC to describe our algorithm.
\begin{align*}
  \textbf{(Dual-VC)} & \quad
  \begin{array}{lll}
    \text{maximize }&\displaystyle \sum_{e\in E}y_e\\
    \text{subject to }&\displaystyle \sum_{e\in\delta(v)}y_e\leq w(v) & \text{for }v\in V, \\
    & y_e \geq 0 & \text{for } e \in E.
  \end{array}
\end{align*}

As a preprocess, we construct a network and its half-integral flow $f^*$ from the given half-integral optimal dual solution $y^*\in \bbR^E$.
A step in our exhaustive search consists of three parts.
\begin{itemize}
\setlength{\itemsep}{0pt}
\item[(I)] From the current graph and the half-integral flow $f^*$, we compute a corresponding
half-integral optimal primal solution $x^* \in \bbR^V$. (Though we have $x^*$ in the beginning, we need this for recursive steps.)
\item[(II)] We find all variables that can be fixed to integers without changing $\lp(G)$, and we transform the current graph $G$ so that the all-half vector is the unique optimal primal solution.
\item[(III)] We pick an arbitrary edge $\{u,v\} \in E$ that is not covered yet, and we go back to part~(I) recursively after setting $x_u = 1$ or $x_v = 1$.
  In order to perform the following steps, we remove covered edges from $G$ and update the half-integral flow $f^*$.
\end{itemize}

We will show that these three parts can be performed in linear time, and the depth of the search tree is bounded by $2k$.
Thus, the total running time of the algorithm is $O(4^k (|V|+|E|))$.

\paragraph{Preprocess:}
From the current graph $G=(V,E)$, we construct a network $(\overline{G}=(\overline{V}\cup\{s,t\},\overline{E}), c)$ as follows.
\begin{align*}
\overline{V}&=\overline{L} \cup \overline{R}, \quad \overline{L}=\{l_v \mid v\in V\}, \quad \overline{R}=\{r_v\mid v\in V\},\\
\overline{E}&=\{(s,l_v) \mid v\in V\} \cup\{(l_u,r_v)\mid \{u,v\}\in E\} \cup \{(r_v,t)\mid v\in V\},\\
c(e)&=\begin{cases}
w(v)& (e=(s,l_v)\text{ or }e=(r_v,t)),\\
\infty& (\text{otherwise}).
\end{cases}
\end{align*}

\begin{proposition}\label{prp:flow-from-dual}
  Given a dual solution $y$, define $f:\overline{E}\to\bbR$ as $f(s,l_v)=f(r_v,t)=\sum_{e\in\delta(v)}y_e$ for $v\in V$, and $f(l_u,r_v)=y_{e}$ for $e = \{u,v\}\in E$.
  Then $f$ is a flow in $\overline{G}$ of amount $2\val(G,y)$.
\end{proposition}
\begin{proof}
  Since $y$ is a feasible dual solution,
  we have $f(s,l_v) = f(r_v,t) = \sum_{e\in\delta(v)}y^*_{e} \leq w(v)$ for any $v \in V$.
  Also, $f$ clearly satisfies the condition $\sum_{e \in \delta^+(v)}f(e) = \sum_{e \in \delta^-(v)}f(e)$ for any $v \in \overline{V}$.
  It is easy to see that the amount of $f$ is $2\val(G,y)$.
\end{proof}

\begin{proposition}\label{prp:dual-from-flow}
  Given a flow $f$ of $\overline{G}$, define a solution $y$ as $y(u,v) = \frac{1}{2}(f(l_u, r_v) + f(l_v, r_u))$ for $e = \{u,v\} \in E.$
  Then $y$ is a feasible dual solution for $G$ and $\val(G,y)$ is half the amount of $f$.
\end{proposition}
\begin{proof}
  Since $f$ is a flow,
  we have $\sum_{\{u,v\} \in \delta(v)}y(e) = \frac{1}{2}\sum_{\{u,v\}\in \delta(v)}(f(l_u,r_v)+f(l_v,r_u)) \leq w(v)$ for any $v \in V$.
  Also, $\val(G,y) = \sum_{e\in E}y(e) = \frac{1}{2}\sum_{\{u,v\}\in E}\left(f(l_u,r_v)+f(l_v,r_u)\right)$, which is half the amount of $f$.
\end{proof}
From the correspondence between a dual solution and a flow, we have the following.
\begin{corollary}\label{crl:max-flow-max-dual}
  From a maximum flow $f^*$ for $\overline{G}$, we can compute an optimal dual solution $y^*$ for $G$ (and vice versa).
\end{corollary}

As a preprocess, we construct an optimal flow $f^*$ from the optimal dual solution $y^*$ using Proposition~\ref{prp:flow-from-dual}.
Note that $f^*$ is also half-integral.

\paragraph{Part (I):}
Given an optimal flow $f^*$, we create a primal solution $x^* \in \bbN^V_{1/2}$ from the residual graph $\overline{G}_{f^*}$ as follows:
\begin{align*}
x^*_v=\begin{cases}
0&(l_v\text{ is reachable from }s\text{ and }r_v\text{ is not reachable from }s\text{ in }\overline{G}_{f^*}),\\
1&(l_v\text{ is not reachable from }s\text{ and }r_v\text{ is reachable from }s\text{ in }\overline{G}_{f^*}),\\
\frac{1}{2}&(\text{otherwise}).
\end{cases}
\end{align*}

\begin{lemma}[\cite{Nemhauser:1975ko}]\label{lem:x*-optimal-primal-solution}
  $x^*$ is an optimal primal solution.
\end{lemma}

\paragraph{Part (II):}
Now we have an optimal primal solution $x^*$ from Lemma~\ref{lem:x*-optimal-primal-solution}.
We denote by $y^*$ the optimal dual solution created from $f^*$.
From Property~(2), we can assume that there is a vertex cover of the minimum weight containing all vertices $v$ with $x^*_v=1$ and no vertices $v$ with $x^*_v=0$.
Let $V'$ be the set of vertices $v$ with $x^*_v  =\frac{1}{2}$ and $E' \subseteq E$ be the set of edges that are not covered by vertices $v$ with $x^*_v=1$.
We define $G'=(V',E')$.
\begin{lemma}
  Let $x'$ and $y'$ be the restriction of $x^*$ and $y^*$ to $V'$ and $E'$, respectively.
  Then, $x'$ and $y'$ are optimal primal and dual solutions for $G'$, respectively.
\end{lemma}
\begin{proof}
  The solution $x'$ must be an optimal primal solution in $G'$ since we have set integers to variables according to $x^*$.

  Let $U \subseteq V$ be the set of vertices $v$ with $x^*_v = 1$.
  Note that $\val(G,x^*) - \val(G',x') = \sum_{v \in U}w(v)$.
  On the other hand,
  $\val(G,y^*) - \val(G',y') = \sum_{e \in E: \text{incident to }U }y^*(e) \leq \sum_{v \in U}w(v)$.
  Since $\val(G,x^*) = \val(G,y^*)$,
  it follows that $\val(G',x') \leq \val(G',y') $.
  Since $x'$ is an optimal primal solution, $y'$ must be an optimal dual solution.
\end{proof}
Note that $x'$ in the lemma above is again the all-half vector in $\bbR^{V'}$.
Now we restrict our attention to $G'$ and replace $G$ by $G'$.
Also, we can recompute the residual graph just by ignoring vertices in $\overline{V}$ corresponding
to vertices in $V \setminus V'$.
Since the restriction of $y^*$ is an optimal dual
solution of $G'$, the restriction of $f^*$ is a maximum flow of $\bar{G'}$.

In order to transform the current graph $G$ so that it admits the all-half vector as its unique optimal LP solution,
we need several properties of the residual graph $\overline{G}_{f^*}$.
To avoid confusion, we use $\overline{\phantom{}\cdot\phantom{}}$ to denote a vertex set in $\overline{V}$, e.g., $\overline{S}$.
For a subset $\overline{S}\subseteq \overline{V}$, we define $S_L=\{v\in V\mid l_v\in \overline{S} \cap \overline{L} \}$ and $S_R=\{v\in V\mid r_v\in \overline{S} \cap \overline{R}\}$.
For a subset $T\subseteq V$, we define $\overline{L}_T=\{l_v \in \overline{V} \mid v\in T\}$ and $\overline{R}_T=\{r_v \in \overline{V} \mid v\in T\}$.

\begin{lemma}\label{lem:residual1}
  For a vertex set $\overline{S} \subseteq \overline{V}$, the following are equivalent:
  \begin{itemize}
  \setlength{\itemsep}{0pt}
    \item[(1)] There is no edge from $\overline{S}$ to $\overline{V}\setminus \overline{S}$ in $\overline{G}_{f^*}$.
    \item[(2)] $N(S_L)=S_R$ and $w(S_L)=w(S_R)$.
  \end{itemize}
  \end{lemma}
\begin{proof}
  In the proof, we use the notation of a neighbor set $N(\cdot)$ for $G$ (neither $\overline{G}$ nor $\overline{G}_{f^*}$).

  Since the all-half vector is an optimal primal solution, the optimal LP value is $\frac{1}{2}\sum_{v\in V}w(v)$.
  However in the dual, $\sum_{e \in E}y_e = \frac{1}{2}\sum_{v \in V}\sum_{e \in \delta(v)}y_e \leq \frac{1}{2}\sum_{v \in V}w(v)$.
  Thus, $\sum_{e \in \delta(v)}y_e = w(v)$ must hold for every $v \in V$.
  Hence, in the flow $f^*$, every edge incident to $s$ or $t$ is saturated.
  This means that the amount of the out-flow of $\overline{S} \cap \overline{L}$ is
  $w(S_L)$ and the amount of the in-flow of $\overline{S} \cap \overline{R}$ is $w(S_R)$.
  We define $\overline{E}_L = \{(l_u,r_v) \in \overline{E} \mid l_u\in \overline{S}, r_v\not\in \overline{S}\}$ and $\overline{E}_R = \{(r_v,l_u) \mid (l_u,r_v) \in \overline{E}, f(l_u,r_v)>0, r_v\in \overline{S}, l_u\not\in \overline{S}\}$.

  $(1)\Rightarrow(2)$:
  Note that no edge from $\overline{L}$ to $\overline{R}$ can be saturated.
  Thus, the set of edges outgoing from $\overline{S}$ to $\overline{V}\setminus \overline{S}$ in $\overline{G}_{f^*}$ is $\overline{E}_L \cup \overline{E}_R$ (Recall that $\overline{V}$ does not contain $s$ and $t$).
  Hence, $\overline{E}_L = \overline{E}_R = \emptyset$ from the assumption.
  We have $N(S_L)\subseteq S_R$ from $\overline{E}_L = \emptyset$.
  Also, we have $w(S_R)\leq w(S_L)$ from $\overline{E}_R = \emptyset$.
  Since the amount of the out-flow of $\overline{S} \cap \overline{L}$ is at most the amount of the in-flow of $\overline{R}_{N(S_L)}$, we have $w(S_L)\leq w(N(S_L))$.
  Therefore, we obtain $w(S_L)\leq w(N(S_L))\leq w(S_R)\leq w(S_L)$, which implies $N(S_L)=S_R$ and
  $w(S_L)=w(S_R)$.

  $(2)\Rightarrow(1)$: Since $N(S_L)=S_R$, $\overline{E}_L$ must be an empty set.
  Moreover, since $w(S_L)=w(S_R)$, all the in-flow of $\overline{R}_{N(S_L)} = \overline{R}_{S_R} = \overline{S} \cap \overline{R}$ must come from $\overline{S} \cap \overline{L}$.
  Therefore $\overline{E}_R$ must be an empty set.
  Thus, there is no edge from $\overline{S}$ to $\overline{V}\setminus \overline{S}$.
\end{proof}

\begin{lemma}\label{lem:residual2}
  If a tail strongly connected component $\overline{S}$ of $\overline{G}_{f^*}$ satisfies $S_L\cap S_R=\emptyset$, then $S_L$ is an independent set of $G$ and $w(S_L)=w(N(S_L))$.
\end{lemma}
\begin{proof}
  Since there is no edge from $\overline{S}$ to $\overline{V}\setminus \overline{S}$,
  $N(S_L)=S_R$ and $w(S_L)=w(S_R)$ hold from Lemma~\ref{lem:residual1}.
  From the former property, $S_L\cap N(S_L)=S_L\cap S_R = \emptyset$, and therefore $S_L$ is an independent set.
  From the both properties, $w(S_L)=w(S_R)=w(N(S_L))$ holds.
\end{proof}

Indeed, the converse is also true.
\begin{lemma}\label{lem:residual3}
  If there is an independent set $S\subseteq V$ of $G$ that satisfies $w(S)=w(N(S))$, then there exists a tail strongly connected component $\overline{T}$ of $\overline{G}_{f^*}$ that satisfies $T_L\cap T_R=\emptyset$.
\end{lemma}
\begin{proof}
  Let $S\subseteq V$ be the minimal independent set that satisfies $w(S)=w(N(S))$.
  We prove that $\overline{T}=\overline{L}_S\cup \overline{R}_{N(S)}$ satisfies the properties above.
  Since $N(T_L)=N(S) = T_R$ and $w(T_L)= w(S) = w(N(S)) = w(T_R)$ hold, from Lemma~\ref{lem:residual1}, there is no edge from $\overline{T}$ to $\overline{V}\setminus \overline{T}$.

  Suppose for contradiction that $\overline{T}$ is not strongly connected.
  Then, there is a subset $\overline{T'}\subsetneq \overline{T}$ such that there is no edge from $\overline{T'}$ to $\overline{T}\setminus \overline{T}'$.
  Since there is no edge from $\overline{T}$ to $\overline{V}\setminus \overline{T}$, this implies there is no edge from $\overline{T'}$ to $\overline{V}\setminus \overline{T'}$.
  Therefore, from Lemma~\ref{lem:residual1}, $w(T'_L)=w(T'_R)=w(N(T'_L))$ holds.
  This contradicts the minimality of $S$ as $T'_L\subseteq S$.
  Thus, $\overline{T}$ must be a tail strongly connected component.
  Since $S$ is an independent set, $T_L\cap T_R=S\cap N(S)=\emptyset$ holds.
\end{proof}

From Property~(4),
as long as there is an independent set $S \subseteq V$ with $w(S) = w(N(S))$,
we can safely assign integers to $S$ and $N(S)$.
That is, we set $x_u = 1$ for all $u \in N(S_L)$ and set $x_u = 0$ for all $u \in S_L$, and we remove them from the graph.
From Property~(5),
if we no longer have such an independent set, the all-half vector becomes the unique optimal primal solution.

From Lemmas~\ref{lem:residual2} and~\ref{lem:residual3},
such an independent set exists if and only if there is a tail strongly connected component $\overline{S}$ in $\overline{G}_{f^*}$ with $S_L \cap S_R = \emptyset$.
Thus, it suffices to keep finding such tail strongly connected components.
An issue here is that (apparently) we have to recompute residual graphs.
Fix an independent set $S$ with $w(S) = w(N(S))$ and let $G'$ be the graph obtained from $G$ by removing vertices in $S \cup N(S)$ and removing edges incident to $N(S)$.
To avoid recomputing the residual graph of $\overline{G'}$ again, we use the following fact.
\begin{lemma}
  The restriction of $x^*$ to $V'$ and the restriction of $y^*$ to $E'$ are optimal primal and dual solutions for $G'$, respectively.
  In particular, since $x^*$ is the all-half vector, the all-half vector in $\bbR^{V'}$ is an optimal primal solution of $G'$.
\end{lemma}
\begin{proof}
  Let $x'\in\bbR^{V'}$ and $y'\in \bbR^{E'}$ be the restrictions of $x^*$ and $y^*$, respectively.
  Since the amount of out-flow from $\overline{S} \cap \overline{L}$ is $w(S_L)$ and $w(S_L) = w(N(S_L))$,
  all out-flow from $\overline{S} \cap \overline{L}$ flows into $\overline{R}_{N(S_L)}$.
  Thus, $\val(G, y^*) - \val(G',y') = w(S_L)$.
  Also, $\val(G, x^*) - \val(G',x') = \frac{1}{2}(w(S_L) + w(N(S_L))) = w(S_L)$.
  Hence $\val(G',x') = \val(G',y')$ holds, and it follows that $x'$ and $y'$ are optimal primal/dual solutions.
\end{proof}
By seeing the construction of the residual graph, we have the following.
\begin{corollary}
  The residual graph of $G'$ is obtained from $\overline{G}_{f^*}$ by removing the tail strongly connected component $\overline{S}$.
\end{corollary}
Thus, what we have to do is keep finding and removing tail strongly connected components $\overline{S}$ with $S_L \cap S_R = \emptyset$ in a fixed residual graph!
We can perform this process in linear time, and in the end we obtain a graph for which the all-half vector is the unique optimal primal solution.

\paragraph{Part (III):}
In this part, we choose an arbitrary edge $\{u,v\}\in E$ that is not covered yet, and invoke recursion after setting $x_u=1$ or $x_v=1$.

Suppose that we have set $x_u = 1$. (The other case is similar.)
We explain how to update the optimal dual solution.
Let $f^*$ be the current flow.
Then, we remove all flow passing through $l_u$ or $r_u$.
Since all edges between $\overline{L}$ and $\overline{R}$ are directed from $\overline{L}$ to $\overline{R}$,
we can remove them in $O(|\delta(u)|)$ time.
After removing $l_u$ and $r_u$ from $\overline{G}$, we augment the flow as long as we can.
If the amount of the augmentation is $\Delta$,
the running time is $O(\Delta(|V|+|E|))$ from the half-integrality of $f^*$.
Thus, the running time in this step is bounded by $O(|\delta(v)| + \Delta(|V|+|E|))$.

Now we calculate the total running time spent in Part (III).
To this end, we see the connection between $\Delta$ and the change of the optimal LP value in one step.
Since we have changed $x_u$ from $\frac{1}{2}$ to $1$, we gain $\frac{w(u)}{2}$.
Note that the amount of the flow decreases by $w(u) - \Delta$.
From Proposition~\ref{prp:dual-from-flow}, we lose $\frac{w(u)}{2} - \frac{\Delta}{2}$.
In total, the LP value increases by $\frac{\Delta}{2}$.
Thus, we can decrease $k$ by $\frac{\Delta}{2}$.

Throughout the algorithm, the sum of $|\delta(v)|$ is bounded by $|E|$ and the sum of $\Delta$ is bounded by $2k$.
Since the number of leaves in the search tree is at most $2^{2k}=4^k$,
the running time caused by $|\delta(v)|$ is at most $O(4^k|E|)$.
Since we spend $O(\Delta(|V|+|E|))$ time to decrease $k$ by $\frac{\Delta}{2}$,
the running time caused by augmenting the flow is at most $O(4^k(|V|+|E|))$.
In total, the running time of Part (III) is $O(4^k(|V|+|E|))$ and this is dominant in the whole algorithm.

\bibliographystyle{abbrv}
\bibliography{paper}

\newpage

\section*{Appendix}
\appendix

\section{List of Binary \BIP Problems}\label{sec:problems}
Many important problems can be formulated as BIP2. Below we see several examples.

\begin{example}[\VC]
  \rm
  Given a graph $G=(V,E)$,
  by introducing a variable $x_v$ that represents whether a vertex $v \in V $ is contained in a vertex cover, the problem can be formulated as Binary \BIP as follows:
  \begin{align*}
    \begin{array}{lll}
      \text{minimize }&\displaystyle \sum_{v\in V}w(v)x_v\\
      \text{subject to: }&x_v+x_u\geq 1 & \text{for }\{u,v\}\in E.\\
    \end{array}
  \end{align*}
\end{example}

\begin{example}[\OCT]
  \rm
  Given a graph $G=(V,E)$,
  by introducing variables $x_v$, $l_v$, and $r_v$ each describing whether a vertex $v$ is contained
  in an odd cycle transversal $S$, in the left side of the bipartite graph $G[V\setminus S]$, and in the right side of $G[V \setminus S]$, respectively,
  the problem can be formulated as Binary \BIP as follows:
  \begin{align*}
    \begin{array}{lll}
      \text{minimize }&\displaystyle \sum_{v\in V}w(v)x_v\\
      \text{subject to: }&l_v+r_v+x_v\geq 1\enspace\text{for }v\in V,\\
      &l_u+l_v\leq 1\enspace\text{for }\{u,v\}\in E,\\
      &r_u+r_v\leq 1\enspace\text{for }\{u,v\}\in E.
    \end{array}
  \end{align*}
\end{example}

\begin{example}[\ASAT]
  \rm
  Given a 2-CNF $\mathcal{C}$ over variables $V$,
  by introducing a variable $x_v$ representing the assigned value of a variable $v$ and a variable $z_C$ representing whether a clause $C$ is unsatisfied, the problem can be formulated as Binary \BIP as
  follows:
  \begin{align*}
    \begin{array}{lll}
      \text{minimize }&\displaystyle \sum_{C\in\mathcal{C}}z_C\\
      \text{subject to: }&x_u+x_v+z_C\geq 1 & \text{for }C=(v\vee u)\in\mathcal{C},\\
      &x_u+(1-x_v)+z_C\geq 1 & \text{for }C=(v\vee\bar{u})\in\mathcal{C},\\
      &(1-x_u)+x_v+z_C\geq 1 & \text{for }C=(\bar{v}\vee u)\in\mathcal{C},\\
      &(1-x_u)+(1-x_v)+z_C\geq 1 & \text{for }C=(\bar{v}\vee\bar{u})\in\mathcal{C}.
    \end{array}
  \end{align*}
\end{example}

Below is a list of other problems that can be written in the form of Binary \BIP.
All these problems can be solved in $O(4^k(n+m))$ time, where $k$ is the solution size.
In graph problems, $n$ and $m$ denote the number of vertices and the number of edges, respectively,
and in problems related to SAT, $n$ and $m$ denote the number of variables and the number of constraints, respectively.

\begin{itemize}
\setlength{\itemsep}{0pt}
\item \textsf{Split Vertex Deletion}:
An undirected graph $G = (V, E)$ is called a \emph{split graph} if there is a
vertex set $C\subseteq V$ such that $G[C]$ is a clique and $G[V\setminus C]$ is an independent set.
Given a weighted undirected graph $G=(V,E)$, \textsf{Split Vertex Deletion}\xspace
is a problem of finding a minimum weight vertex set $S\subseteq V$ such that
$G[V\setminus S]$ becomes a split graph.
\item \textsf{Edge Bipartization:}
Given an edge-weighted undirected graph $G=(V,E)$, \textsf{Edge Bipartization}\xspace is
a problem of finding a minimum weight edge set $S \subseteq E$ such that the
graph $G' = (V, E\setminus S)$ becomes bipartite.
\item \textsf{Min SAT:}
Given a CNF $\mathcal{C}$ over variables $V$, \MINSAT is a problem of finding a Boolean assignment
that minimizes the number of satisfied clauses.
\item \textsf{Generalized Vertex Cover:}
The instance of \textsf{Generalized Vertex Cover}\xspace
is a weighted undirected graph $G=(V,E)$ together with three edge weight
functions $d_0, d_1, d_2$ that satisfy $d_0(e) \geq d_1(e) \geq d_2(e) \geq 0$
for every edge $e\in E$.
\textsf{Generalized Vertex Cover}\xspace is a problem of finding a vertex set
$S\subseteq V$ that minimizes the weight $\sum_{v\in S}w(v) + \sum_{e\in
E}d_{|e\cap S|}(e)$.
\item \textsf{Generalized 2-SAT:}
For a given 2-CNF formula $\mathcal{C}$ over variables $V$ and a function $w:V
\to \bbN$, \textsf{Generalized 2-SAT}\xspace is a problem of finding a Boolean
assignment that makes the formula true and minimizes the total cost of variables that are assigned true.
\item \textsf{Complement of Maximum Clique:}
For a given graph, \textsf{Complement of Maximum Clique}\xspace is a problem of
finding a minimum weight vertex set $S$ such that $G[V\setminus S]$ becomes a clique.
\item \textsf{Almost Boolean 2-CSP:}
Given a Boolean 2-CSP $\mathcal{C}$ over variables $V$,
\textsf{Almost Boolean 2-CSP}\xspace is a problem of finding a Boolean
assignment that minimizes the number of unsatisfied constraints.
\item \textsf{Directed Min UnCut:}
Given an edge-weighted directed graph $G=(V,E)$,
\textsf{Directed Min UnCut}\xspace is a problem of finding a partition $(S,V\setminus S)$ that minimizes the number of edges not going from $S$ to $V \setminus S$.
\end{itemize}

\section{Linear-Time Reduction to \VCLP}\label{sec:reduction}
Let $\calI$ be an \BIP instance.
An \emph{LP solution pair} of $\calI$ refers to a pair of primal/dual solutions for $\calI$.
We say that an LP solution pair is \emph{optimal} if both the primal solution and the dual solution are optimal,
and is \emph{half-integral} if both the primal solution and the dual solution are half-integral.
In this section, we prove Theorem~\ref{thm:reduction} by showing a linear-time algorithm that, given a \BIP instance $\calI$ and its half-integral optimal LP solution pair, outputs a vertex cover instance $G$ with $\gap(G) = \gap(\calI)$ and its half-integral optimal LP solution pair.

Our algorithm consists of a sequence of reductions and we have mainly three parts.
\begin{itemize}
\item[(I)] We transform the instance so that each constraint has the form of $x_i+x_j+z_{i,j}\geq c_{i,j}$, i.e., the coefficient of every shared variable is one and every constraint has an independent variable.
\item[(II)] We remove all independent variables.
\item[(III)] We construct an instance of \VC by restricting the domain of variables to be binary.
\end{itemize}

\paragraph{Part (I).}
Let $\calI$ be an instance of $\BIP$.
We simplify $\calI$ so that each constraint has an independent variable as follows.
For each constraint of the form $a_{i,j}x_i+b_{i,j}x_j\geq c_{i,j}$,
we create an independent variable $z_{i,j}$ of sufficiently large weight and replace the constraint with $a_{i,j}x_i+b_{i,j}x_j+z_{i,j}\geq c_{i,j}$.
Let $\calI^1$ be the resulting instance.
\begin{lemma}
  We have $\gap(\calI^1) = \gap(\calI)$.
  From a (half-integral) optimal LP solution pair for $\calI$,
  we can compute a (half-integral) optimal LP solution pair for $\calI^1$ in linear time.
\end{lemma}
\begin{proof}
  Note that $z_{i,j}$ must be zero in any optimal LP/IP solution.
  Thus, the modification does not change optimal LP/IP values,
  and we can directly use the given optimal LP solution pair for $\calI$ as the one for $\calI^1$.
\end{proof}

Now we partition the constraint set $E$ to $E^{+}$, $E^{\neq}$, and $E^{-}$ as follows.
For each constraint of the form $a_{i,j}x_i + b_{i,j} + z_{i,j} \geq c_{i,j}$,
we add it to $E^{+}$ if $(a_{i,j}, b_{i,j}) = (1,1)$,
add it to $E^{\neq}$ if $(a_{i,j},b_{i,j}) = (1,-1)$,
and add it to $E^{-}$ if $(a_{i,j},b_{i,j}) = (-1,-1)$.
(If $(a_{i,j},b_{i,j}) = (-1,1)$, then the constraint is added to $E^{\neq}$ after swapping $i$ and $j$.)
We regard $E^+$ and $E^-$ as sets of unordered pairs and $E^{\neq}$ as a set of ordered pairs.

Though all LP variables in this section are constrained to be non-negative,
in what follows, we omit them from LP formalizations for brevity.
The LP relaxation of $\calI^1$ and its dual can be written as follows.
\begin{align*}
  \textbf{(Primal-1)} & \quad
  \begin{array}{lll}
    \text{minimize }&\displaystyle \sum_{i \in V}w_i x_i + \sum_{(i,j)\in E}d_{i,j}z_{i,j}\\
    \text{subject to }&x_i+x_j+z_{i,j}\geq c_{i,j}&\text{for }(i,j)\in E^{+},\\
      &x_i-x_j+z_{i,j}\geq c_{i,j}&\text{for }(i,j)\in E^{\neq},\\
    	&-x_i-x_j+z_{i,j}\geq c_{i,j}&\text{for }(i,j)\in E^{-}.\\
  \end{array} \\
  \textbf{(Dual-1)} & \quad
  \begin{array}{lll}
    \text{maximize }&\displaystyle \sum_{(i,j)\in E}c_{i,j}y_{i,j}\\
    \text{subject to }&\displaystyle \sum_{(i,j)\in E^{+} \cup E^{\neq}}y_{i,j}-\sum_{(j,i)\in E^{\neq}\cup E^-}y_{j,i}\leq w_i &
    \text{for }i\in V,\\
    &y_{i,j}\leq d_{i,j}&\text{for }(i,j)\in E. \\
  \end{array}
\end{align*}

Now, we want to transform the instance so that each constraint has the form of $x_i+x_j+z_{i,j}\geq c_{i,j}$, i.e., $a_{i,j}=b_{i,j}=1$ for every $(i,j) \in E$.
Let $X$ and $M$ be sufficiently large integers.
For a variable $x_i$ of weight $w_i$, we create two variables $x^+_i$ and $x^-_i$.
We set the weight of $x^+_i$ to be $w^+_i=M+w_i$ and the weight of $x^-_i$ to be $w^-_i=M$.
For each constraint, we replace $+x_i$ with $x^+_i$ and $-x_i$ with $x^-_i-X$.
Finally, we add a new constraint of the form $x^+_i+x^-_i\geq X$.
Let $\calI^2$ be the resulting instance.

The primal and dual LP relaxations of $\calI^2$ can be written as follows:
\begin{align*}
  \textbf{(Primal-2)} & \quad
  \begin{array}{lll}
    \text{minimize }&\displaystyle \sum_{i \in V} ((M+w_i) x^+_i+M x^-_i) + \sum_{(i,j)\in
    E}d_{i,j}z_{i,j}\\
    \text{subject to }&x^+_i+x^+_j+z_{i,j}\geq c_{i,j}&\text{for }(i,j)\in E^+,\\
      &x^+_i+x^-_j+z_{i,j}\geq c_{i,j}+X&\text{for }(i,j)\in E^{\neq},\\
      &x^-_i+x^-_j+z_{i,j}\geq c_{i,j}+2X&\text{for }(i,j)\in E^-,\\
      &x^+_i+x^-_i\geq X&\text{for }i\in V,\\
  \end{array} \\
  \textbf{(Dual-2)} & \quad
  \begin{array}{lll}
    \text{maximize }&\multicolumn{2}{l}{\displaystyle \sum_{(i,j)\in E}c_{i,j}y_{i,j}+X\sum_{(i,j)\in E^{\neq}} y_{i,j}+2X\sum_{(i,j)\in E^-} y_{i,j}+X\sum_{i \in V}  \gamma_i}\\
    \text{subject to }&\displaystyle \sum_{(i,j)\in E^+\cup E^{\neq}}y_{i,j} +\gamma_i\leq M+w_i & \text{for }i\in V,\\
    &\displaystyle \sum_{(j,i)\in E^{\neq}\cup E^-}y_{j,i}+\gamma_i\leq M & \text{for }i\in V,\\
    &y_{i,j}\leq d_{i,j}&\text{for }(i,j)\in E.\\
  \end{array}
\end{align*}

\begin{lemma}
  We have $\gap(\calI^2) = \gap(\calI^1)$.
  Furthermore, 
  from a (half-integral) optimal LP solution pair for $\calI^1$,
  we can compute a (half-integral) optimal LP solution pair for $\calI^2$ in linear time.
\end{lemma}
\begin{proof}
  Because $M$ is sufficiently large, for any optimal primal/IP solution of $\calI^2$, the constraint of the form $x^+_i+x^-_i \geq X$ must be satisfied with equality.
  Therefore, we can assume that $x^-_i=X-x^+_i$ holds in any optimal primal/IP solution.
 
  For any primal/IP solution $x \in \bbR^V$ of $\calI^1$, 
  we can obtain a primal/IP solution of $\calI^2$ by setting $x^+_i=x_i$ and $x^-_i=X-x_i$, 
  and for any primal/IP solution $(x^+,x^-)$ of $\calI^2$, we can obtain a primal/IP solution of $\calI^1$ by setting $x_i=x^+_i$.
  Thus, optimal LP/IP values increase exactly by $\sum_{i \in V}((M+w_i)x_i+M(X-x_i)-w_i x_i)=|V| X M$.
  Hence, $\gap(\calI^2) = \gap(\calI^1)$ and we can compute the optimal primal solution for $\calI^2$ in linear time.

  We can obtain an optimal dual solution $(y,\gamma)$ of $\calI^2$ from an optimal dual solution $y^*$ of $\calI^1$ by setting $y = y^*$ and 
  \[
    \gamma_i=M-\sum_{(j,i)\in E^{\neq} \cup E^-}y_{j,i}.
  \]
  We can easily check that it is feasible and its value is exactly $\lp(\calI^2)$.
\end{proof}
For readability, we restate \textbf{(Primal-2)} and \textbf{(Dual-2)} as follows.
\begin{align*}
  \textbf{(Primal-2')} & \quad
  \begin{array}{lll}
    \text{minimize }&\displaystyle \sum_{i \in V} w_i x_i + \sum_{\{i,j\}\in E}d_{i,j}z_{i,j}\\
    \text{subject to }&x_i+x_j+z_{i,j}\geq c_{i,j}&\text{for }\{i,j\}\in E,\\
  \end{array} \\
  \textbf{(Dual-2')} & \quad
  \begin{array}{lll}
    \text{maximize }&\displaystyle \sum_{\{i,j\}\in
    E}c_{i,j}y_{i,j}\\
    \text{subject to }&\displaystyle \sum_{\{i,j\}\in E}y_{i,j}\leq w_i& \text{for }i\in V,\\
    &y_{i,j}\leq d_{i,j}&\text{for }\{i,j\}\in E.
  \end{array}
\end{align*}

\paragraph{Part (II).}
We want to remove independent variables from an instance $\calI^2$ of the form \textbf{(Primal-2')}.
For a constraint of the form $x_i+x_j+z_{i,j}\geq c_{i,j}$ with an independent variable $z_{i,j}$ of weight
$d_{i,j}$, we create two variables $z^i_{i,j}$ and $z^j_{i,j}$ of weight $d_{i,j}$.
Then, we replace the constraint with the following three constraints:
\begin{align*}
  x_i+z^i_{i,j} \geq c_{i,j}, \quad x_j+z^j_{i,j} \geq c_{i,j}, \quad z^i_{i,j}+z^j_{i,j} \geq c_{i,j}.
\end{align*}
Let $\calI^3$ be the resulting instance.

Let $y^i_{i,j}$, $y^j_{i,j}$, and $y^z_{i,j}$ be dual variables corresponding to the newly added three constraints.
Then, the primal and dual LP relaxations of $\calI^3$ can be written as follows:
\begin{align*}
  \textbf{(Primal-3)} & \quad
  \begin{array}{lll}
    \text{minimize }&\displaystyle \sum_{i \in V} w_i x_i + \sum_{\{i,j\}\in
    E}d_{i,j}(z^i_{i,j}+z^j_{i,j})\\
    \text{subject to }&x_i+z^i_{i,j}\geq c_{i,j}&\text{for }\{i,j\}\in E,\\
      &x_j+z^j_{i,j}\geq c_{i,j}&\text{for }\{i,j\}\in E,\\
      &z^i_{i,j}+z^j_{i,j}\geq c_{i,j}&\text{for }\{i,j\}\in E,\\
  \end{array} \\
  \textbf{(Dual-3)} & \quad
  \begin{array}{lll}
    \text{maximize }&\multicolumn{2}{l}{\displaystyle \sum_{\{i,j\}\in
    E}c_{i,j}(y^i_{i,j}+y^j_{i,j}+y^z_{i,j})}\\
    \text{subject to }&\displaystyle \sum_{\{i,j\}\in E}y^i_{i,j}\leq w_i& \text{for }i\in V,\\
    &y^i_{i,j}+y^z_{i,j}\leq d_{i,j}&\text{for }\{i,j\}\in E.
  \end{array}
\end{align*}

\begin{lemma}
  We have $\gap(\calI^3) = \gap(\calI^2)$.
  Furthermore, 
  from a (half-integral) optimal LP solution pair for $\calI^2$,
  we can compute a (half-integral) optimal LP solution pair for $\calI^3$ in linear time.
\end{lemma}
\begin{proof}
  For any primal/IP solution $x \in \bbR^V$ of $\calI^2$, we can obtain a primal/IP solution of the $\calI^3$ by setting $z^i_{i,j}=\max(0,c_{i,j}-x_i)$ and $z^j_{i,j}=\max(c_{i,j}-z^i_{i,j},c_{i,j}-x_j)$.
  This is feasible since $x_i+z^i_{i,j} = \max(x_i + c_{i,j})\geq c_{i,j}$,
  $x_j+z^j_{i,j} = \max(x_j + c_{i,j})\geq c_{i,j}$, and $z^i_{i,j}+z^j_{i,j}=c_{i,j}+\max(0,c_{i,j}-x_i-x_j)\geq c_{i,j}$.

  Suppose we have an optimal primal/IP solution of $\calI^3$.
  Since $z_{i,j}$ appears only in one constraint, we can assume that, in the optimal primal/IP solution, $z_{i,j}=\max(0,c_{i,j}-x_i-x_j)$ holds.
  Therefore, $z^i_{i,j}+z^j_{i,j}=z_{i,j}+c_{i,j}$ holds.
  Since $z^i_{i,j}$ and $z^j_{i,j}$ appear only in these constraints,
  in the optimal primal/IP solution, 
  we can assume that $z^i_{i,j}=\max(0,c_{i,j}-x_i)$ and $z^j_{i,j}=\max(c_{i,j}-z^i_{i,j},c_{i,j}-x_j)$ hold.
  Then, we can obtain a solution for $\calI^2$ by setting $z_{i,j}=z^i_{i,j}+z^j_{i,j}-c_{i,j}$.
  This is feasible because $x_i+x_j+z_{i,j}=x_i+x_j+z^i_{i,j}+z^j_{i,j}-c_{i,j}\geq c_{i,j}$.

  Thus, optimal LP/IP values increase by $\sum_{\{i,j\}\in E}d_{i,j}(z^i_{i,j}+z^j_{i,j}-z_{i,j})=\sum_{\{i,j\}\in E}d_{i,j}c_{i,j}$.
  Hence, $\gap(\calI^3) = \gap(\calI^2)$ and we can compute the optimal primal solution for $\calI^3$ in linear time.

  We can obtain a dual optimal solution $(y^i,y^j,y^z)$ of $\calI^3$ from the given dual optimal solution $y$ of $\calI^2$ as follows:
  \begin{align*}
    y^i_{i,j} =y^j_{i,j}=y_{i,j}, \quad  y^z_{i,j}&=d_{i,j}-y_{i,j}.
  \end{align*}
  We can easily check that the obtained dual solution is feasible and has the same value as the primal solution obtained above.
\end{proof}

For readability, we restate \textbf{(Primal-3)} and \textbf{(Dual-3)} as follows.
\begin{align*}
  \textbf{(Primal-3')} & \quad
  \begin{array}{lll}
    \text{minimize }&\displaystyle \sum_{i \in V} w_i x_i\\
    \text{subject to }&x_i+x_j\geq c_{i,j}&\text{for }\{i,j\}\in E,\\
  \end{array} \\
  \textbf{(Dual-3')} & \quad
  \begin{array}{lll}
    \text{maximize }&\displaystyle \sum_{\{i,j\}\in
    E}c_{i,j}y_{i,j}\\
    \text{subject to }&\displaystyle \sum_{\{i,j\}\in E}y_{i,j}\leq w_i& \text{for }i\in V.
  \end{array}
\end{align*}

\paragraph{Part (III).}
Let $\calI^3$ be an instance of the form \textbf{(Primal-3')}.
Finally, we reduce it to an instance of \VC by restricting the domain of variables to be binary.
We use the following lemma.

\begin{lemma}\label{lem:fixing}
  For any half-integral optimal primal solution $x^L \in \bbN^V_{1/2}$ of $\calI^3$, there exists an optimal IP solution $x \in \bbN^V$ that satisfies the following:
  \begin{align*}
    \begin{array}{ll}
    x_i=x^L_i & (x^L_i\in\mathbb{N}),\\
    \lfloor x^L_i\rfloor\leq x_i\leq \lceil x^L_i \rceil & (\text{otherwise}).
    \end{array}
  \end{align*}
\end{lemma}
\begin{proof}
Let $x^I$ be an optimal IP solution.
We construct an optimal IP solution that satisfies the property above.
We define the set of variables as follows:
\begin{align*}
& X_I=\{i \in V \mid x^L_i\in\mathbb{N}\}, \quad
X_H=\{i \in V \mid x^L_i\not\in\mathbb{N}\},\\
& X_<=\{i \in V \mid x^L_i<x^I_i\}, \quad
X_==\{i \in V \mid x^L_i=x^I_i\}, \quad
X_>=\{i \in V \mid x^L_i>x^I_i\}.
\end{align*}
Let $w^L=\lp(\calI^3)$ and $w^I=\opt(\calI^3)$.
For a set of variables $S$, we define $w(S)=\sum_{i\in S}w_i$.
Let $w_I=w(X_I\cap X_<)-w(X_I\cap X_>)$ and $w_H=w(X_H\cap X_<)-w(X_H\cap X_>)$.

First, we construct an IP solution $x^1 \in \bbN^V$ by rounding the non-integral part of $x^L$ towards
$x^I$:
\begin{align*}
x^1_i=\begin{cases}
x^L_i&(i\in X_I),\\
x^L_i+\frac{1}{2}&(i\in X_H\cap X_<),\\
x^L_i-\frac{1}{2}&(i\in X_H\cap X_>).
\end{cases}
\end{align*}
We can check the feasibility of $x^1$ by considering the following three cases.
Fix a constraint $x_i+x_j\geq c_{i,j}$.
If $i,j\in X_I$, then $x^1_i+x^1_j=x^L_i+x^L_j\geq
c_{i,j}$.
If $i\in X_I$ and $j\in X_H$, then $x^1_i+x^1_j\geq x^L_i+x^L_j-\frac{1}{2}\geq c_{i,j}$ since
$x^L_i+x^L_j$ is not an integer.
If $i,j\in X_H$, then $x^1_i+x^1_j\geq\min(x^L_i+x^L_j,x^I_i+x^I_j)\geq c_{i,j}$.
We have $\val(\calI^3,x^1) = w^L+\frac{1}{2}w_H$.
From the optimality of $x^I$, we obtain
\begin{align}\label{eq:w1}
  w^L+\frac{1}{2}w_H\geq w^I.
\end{align}

Then, we construct another IP solution $x^2 \in \bbN^V$ by shifting $x^L$ towards $x^I$:
\begin{align*}
  x^2_i=\begin{cases}
    x^L_i&(i\in X_=),\\
    x^L_i+1&(i\in X_I\cap X_<),\\
    x^L_i-1&(i\in X_I\cap X_>),\\
    x^L_i+\frac{1}{2}&(i\in X_H\cap X_<),\\
    x^L_i-\frac{1}{2}&(i\in X_H\cap X_>).
  \end{cases}
\end{align*}
The feasibility of $x^2$ can be checked by considering the following two cases.
Fix a constraint $x_i+x_j\geq c_{i,j}$.
If $i\in X_I$ and $j\in X_H$, then $x^2_i+x^2_j\geq\min(x^L_i+x^L_j-\frac{1}{2},x^I_i+x^I_j)\geq c_{i,j}$.
Otherwise, $x^2_i+x^2_j\geq\min(x^L_i+x^L_j,x^I_i+x^I_j)\geq c_{i,j}$.
We have $\val(\calI^3,x^2) = w^L+\frac{1}{2}w_H+w_I$.
From the optimality of $x^I$, we obtain
\begin{align}\label{eq:w2}
  w^L+\frac{1}{2}w_H+w_I\geq w^I.
\end{align}

Finally, we construct a half-integral primal solution $x^3 \in \bbN^V_{1/2}$ by shifting $x^I$ towards $x^L$:
\begin{align*}
x^3_i=\begin{cases}
x^I_i&(i\in X_=),\\
x^I_i-\frac{1}{2}&(i\in X_<),\\
x^I_i+\frac{1}{2}&(i\in X_>).
\end{cases}
\end{align*}
Since $x^3_i+x^3_j\geq\min(x^L_i+x^L_j,x^I_i+x^I_j)\geq c_{i,j}$ holds, $x^3$ is feasible.
We have $\val(\calI^3, x^3) = w^I-\frac{1}{2}w_H-\frac{1}{2}w_I$.
From the optimality of $x^L$, we obtain
\begin{align}\label{eq:w3}
  w^I-\frac{1}{2}w_H-\frac{1}{2}w_I\geq w^L.
\end{align}

From~\eqref{eq:w1}+\eqref{eq:w3}, we obtain $w_I\leq 0$, and from~\eqref{eq:w2}+\eqref{eq:w3},
we obtain $w_I\geq 0$.
Therefore, $w_I=0$ holds.
Then, from~\eqref{eq:w1} and~\eqref{eq:w3}, we obtain $w^L+\frac{1}{2}w_H=w^I$.
Thus, $x^1$ is an optimal IP solution that satisfies the given property.
\end{proof}

Now we show how to construct an instance $G$ of \VC.
Let $x^*$ be a half-integral optimal primal solution for $\calI^3$.
For each variable $x_i$, we replace it with $(\lfloor x^*_i\rfloor + x_i)$.
A constraint $x_i+x_j\geq c_{i,j}$ that is satisfied in $x^*$ with strict inequality is always satisfied.
Thus, only constraints that is satisfied with equality remain.
If it consists of two integral variables, $c_{i,j}-\lfloor x^*_i\rfloor-\lfloor x^*_j\rfloor=0$ holds,
and if it consists of two half-integral variables, $c_{i,j}-\lfloor x^*_i\rfloor-\lfloor x^*_j\rfloor=1$ holds.
Let $E'$ be the set of the latter type of constraints.
Then, the primal and dual LP relaxations can be written as follows:
\begin{align*}
  \textbf{(Primal-VC')} & \quad
  \begin{array}{lll}
    \text{minimize }&\displaystyle \sum_{i \in V} w_i \lfloor x^*_i\rfloor + \sum_{i \in V} w_i x_i\\
    \text{subject to }&x_i+x_j\geq 1&\text{for }\{i,j\}\in E',\\
  \end{array} \\
  \textbf{(Dual-VC')} & \quad
  \begin{array}{lll}
    \text{maximize }&\displaystyle \sum_{i \in V} w_i \lfloor x^*_i\rfloor+\sum_{\{i,j\}\in
    E'}y_{i,j}\\
    \text{subject to }&\displaystyle \sum_{\{i,j\}\in E'}y_{i,j}\leq w_i& \text{for }i\in V.
  \end{array}
\end{align*}
Note that we do not need an upper bound on $x_i$ since it always has a binary value in an optimal primal/IP solution.
The primal problem $G$ is \VC with the vertex set $V$ and the edge set $E'$.

\begin{lemma}
  We have $\gap(G) = \gap(\calI^3)$.
  Furthermore, 
  from a half-integral optimal LP solution pair for $\calI^3$,
  we can compute a half-integral optimal LP solution pair for $G$ in linear time.
\end{lemma}
\begin{proof}
  It is clear that the optimal LP value does not change.
  Also from Lemma~\ref{lem:fixing}, it does not change the optimal IP value.
  Thus, we have $\gap(G) = \gap(\calI^3)$ and we can compute a (half-integral) optimal primal solution for $G$ in linear time.

  Let $y^*$ be the given half-integral optimal dual solution.
  It is obviously feasible in the reduced problem.
  From the complementary slackness, for each variable $x_i>0$, $\sum_{\{i,j\}\in E}y^*_{i,j}=w_i$ holds, 
  and for each constraint $x_i+x_j>c_{i,j}$, $y^*_{i,j}=0$ holds.
  Therefore, 
  \begin{align*}
  & \val(G,y^*) =  \sum_{i \in V} w_i \lfloor x^*_i\rfloor+\sum_{\{i,j\}\in E'}y^*_{i,j} 
  = \sum_{i \in V}\sum_{\{i,j\} \in E}y_{i,j} \lfloor x^*_i\rfloor+\sum_{\{i,j\}\in E'}y^*_{i,j} \quad (\text{complementary slackness})\\
  =& \sum_{\{i,j\}\in E}y^*_{i,j}(\lfloor x^*_i\rfloor + \lfloor x^*_j\rfloor)+\sum_{\{i,j\}\in E'}y^*_{i,j} \\
  =& \sum_{\{i,j\}\in E}y^*_{i,j}(\lfloor x^*_i\rfloor + \lfloor x^*_j\rfloor)+\sum_{\{i,j\}\in E'}(c_{i,j}-\lfloor x^*_i\rfloor-\lfloor x^*_j\rfloor)y^*_{i,j} \quad (c_{i,j} - \lfloor x^*_i\rfloor-\lfloor x^*_j\rfloor = 1 \text{ for } (i,j) \in E') \\
  =&\sum_{\{i,j\}\in E}y^*_{i,j}(\lfloor x^*_i\rfloor + \lfloor x^*_j\rfloor)+\sum_{\{i,j\}\in E}(c_{i,j}-\lfloor x^*_i\rfloor-\lfloor x^*_j\rfloor)y^*_{i,j} 
   \quad (c_{i,j} - \lfloor x^*_i\rfloor-\lfloor x^*_j\rfloor = 0 \text{ for } (i,j) \in E \setminus E') \\
  =&\sum_{\{i,j\}\in E}c_{i,j}y^*_{i,j}.
  \end{align*}
  Thus, $y^*$ is a (half-integral) optimal dual solution in $G$.  
\end{proof}

As a result, an instance $\calI$ of \BIP is reduced to an instance $G$ of \VC with $\gap(G) = \gap(\calI)$.
Moreover, we can compute a half-integral optimal LP solution pair for $G$ from a half-integral optimal LP solution pair for $\calI$ in linear time.

\section{A Linear-Time FPT Algorithm for \NMC}\label{sec:multiway}
In this section, we give an $O(4^k(|V|+|E|))$-time FPT algorithm for \NMC.
Let $G=(V,E)$ be an undirected graph and $T\subseteq V$ be a terminal set.
In \NMC, the objective is to find the minimum set of vertices $S$ so that every pair of terminals becomes separated by removing $S$.
We use the size of the optimal solution as a parameter.

For a terminal $t\in T$, we call a subset $C_t$ of vertices an
\emph{isolating cut} of $t$ if its removal separates $t$ from $T\setminus\{t\}$.
We denote the connected component containing $t$ after removing an isolating cut $C_t$ by
$R(C_t)$.
Xiao~\cite{Xiao10} proved the following lemma.
\begin{lemma}[\cite{Xiao10}]\label{lem:isolating}
Let $C_t$ be a minimum isolating cut of a terminal $t$ in a graph $G$.
Then there exists a minimum multiway cut $C$ of $G$ that contains no vertices of
$R(C_i)$.
\end{lemma}

A \emph{farthest minimum isolating cut} of $t$ is an isolating cut $C_t$ that has the largest
$R(C_t)$ among all isolating cuts of the minimum cardinality.
The farthest minimum isolating cut $C_t$ is unique and can be computed in $O(|C_t|(|V|+|E|))$ time
by a single maximum flow computation from $t$ to $T\setminus\{t\}$.
Moreover, if we are given the maximum flow, it can be computed in $O(|V|+|E|)$ time~\cite{picard80}.

Now, we describe our algorithm for \NMC.
First, if the number of terminals $|T|$ is at most 1, we return YES.
Otherwise, we choose an arbitrary terminal $t\in T$.
If $t$ is adjacent to $T\setminus\{t\}$, there is no multiway cut.
If $t$ is already separated from $T\setminus\{t\}$, we remove the set of vertices that are reachable
from $t$ and choose another terminal.
Otherwise, we compute the farthest minimum isolating cut $C_t$ of $t$ by computing a maximum flow
from $t$ to $T\setminus\{t\}$.
If the size of the cut (and thus the amount of the flow) exceeds $k$, we immediately return NO.
We keep and reuse the maximum flow for future updates.
From Lemma~\ref{lem:isolating}, there exists a minimum multiway cut that does not contain any
vertices of $V(R(C_t))$, thus we can safely contract $V(R(C_t))$ to $t$.

After the contraction of $V(R(C_t))$, the farthest minimum isolating cut becomes $N(t)$.
We choose an arbitrary vertex $v\in N(t)$ and branch into two cases, i.e., the minimum multiway cut
contains $v$ or not.
For the former case, we include $v$ into the output and remove it from the graph.
The farthest minimum isolating cut changes to $C_t \setminus\{v\}$ and we can update the maximum
flow by simply removing the flow passing throw $v$.
For the latter case, we contract $v$ to $t$.
We update the maximum flow by searching augmenting paths in the residual graph, recompute
the farthest minimum isolating cut $C_t$ of $t$, and then contract $V(R(C_t))$ to $t$.
Then, we apply the algorithm recursively.

Finally, we analyze the running time of the algorithm.
For simplicity, we fix the ordering $t_1,\ldots,t_{|T|}$ of handling terminals.
Let $T_i(N,\ell,\lambda)$ be the running time of the algorithm to handle $t_i$.
where $N$ is the sum of the number of vertices and edges, $\ell$ is the size of a multiway cut we want to find, and $\lambda$ is the size of
the minimum isolating cut of $t_i$ in the current graph.
Then, we want to bound $T_1(|V|+|E|,k,\lambda_1)$, where $k$ is the given parameter on the solution
size and $\lambda_1$ is the size of the minimum isolating cut of $t_1$ in the input graph $G$.
For notational simplicity, we define $T_{|T|+1}(\cdot,\cdot,\cdot) = 0$.

Suppose $\lambda=0$.
Then, we can find the set $R$ of vertices reachable from $t_i$ in $O(N')$ time, where $N'$ is the sum of the numbers of vertices and edges in $R$.
Let $\lambda'$ be the size of the minimum isolating cut for the next terminal $t_{i+1}$.
Then, we can compute the maximum flow in $O(\lambda' (N-N'))$ time.
Therefore, $T_i(N,\ell,0)\leq T_{i+1}(N-N',\ell,\lambda')+O(N'+\lambda'(N-N'))$ holds.

Suppose $\lambda>0$.
We can compute the farthest minimum isolating cut in $O(N)$ time from the given maximum flow.
Let $v\in N(t)$ be the chosen vertex.
If we include $v$ into the output, $\ell$ and $\lambda$ decrease by $1$.
In this case, we can update the maximum flow in $O(N)$ time.
If we contract $v$ to $t$, $\ell$ remains the same but $\lambda$ increases by at least one because
$C_t=N(t)$ is the unique minimum isolating cut of $t$.
Let $\Delta$ be the increase of $\lambda$.
In this case, we can update the maximum flow by finding augmenting paths $\Delta$ times.
Therefore, $T_i(N,\ell,\lambda)\leq T_i(N,\ell-1,\lambda-1)+T_i(N,\ell,\lambda+\Delta)+O(\Delta N)$ holds.

From the above two recurrences of $T_i$, we can obtain $T_i(N,\ell,\lambda)=O(2^{2\ell-\lambda}N)$.
Thus the total running time is bounded by $O(2^{2k-\lambda_1}(|V|+|E|)) = O(4^k(|V|+|E|))$.

\end{document}